\documentclass[conference]{IEEEtran}
\usepackage[cmex10]{amsmath}
\usepackage{amssymb}

\allowdisplaybreaks[3]

\newtheorem{theorem}{Theorem}
\newtheorem{proposition}[theorem]{Proposition}
\newtheorem{lemma}[theorem]{Lemma}
\newtheorem{definition}[theorem]{Definition}
\newtheorem{corollary}[theorem]{Corollary}
\newtheorem{example}[theorem]{Example}

\begin{document}
\title{Channel Polarization on $q$-ary Discrete Memoryless Channels by Arbitrary Kernels}


\author{
\IEEEauthorblockN{Ryuhei Mori}
\IEEEauthorblockA{Graduate School of Informatics\\
Kyoto University \\
Kyoto, 606--8501, Japan\\
Email: rmori@sys.i.kyoto-u.ac.jp}
\and
\IEEEauthorblockN{Toshiyuki Tanaka}
\IEEEauthorblockA{Graduate School of Informatics\\
Kyoto University\\
Kyoto, 606--8501, Japan\\
Email: tt@i.kyoto-u.ac.jp}
}

\maketitle

\begin{abstract}
A method of channel polarization, proposed by Ar{\i}kan, allows us to construct efficient capacity-achieving channel codes.
In the original work, binary input discrete memoryless channels are considered.
A special case of $q$-ary channel polarization is considered by {\c S}a{\c s}o{\u g}lu, Telatar, and Ar{\i}kan.
In this paper, we consider more general channel polarization on $q$-ary channels.
We further show explicit constructions using Reed-Solomon codes, on which asymptotically fast channel polarization is induced.
\end{abstract}


\IEEEpeerreviewmaketitle

\section{Introduction}
Channel polarization, proposed by Ar{\i}kan, is a method of constructing capacity achieving codes
with low encoding and decoding complexities~\cite{5075875}.
Channel polarization can also be used to construct lossy source codes which achieve rate-distortion trade-off
with low encoding and decoding complexities~\cite{korada2009pco}.
Ar{\i}kan and Telatar derived the rate of channel polarization~\cite{5205856}.
In~\cite{tanaka2010rre}, a more detailed rate of channel polarization which includes coding rate is derived.
In~\cite{5075875}, channel polarization is based on a $2\times 2$ matrix.
Korada, {\c S}a{\c s}o{\u g}lu, and Urbanke considered generalized polarization phenomenon which is based on an $\ell\times\ell$ matrix
and derived the rate of the generalized channel polarization~\cite{korada2009pcc}.
In~\cite{sasoglu2009pad}, a special case of channel polarization on $q$-ary channels is considered.
In this paper, we consider channel polarization on $q$-ary channels which is based on arbitrary mappings.

\section{Preliminaries}
Let $u_0^{\ell-1}$ and $u_i^j$ denote a row vector $(u_0,\dotsc,u_{\ell-1})$ and its subvector $(u_i,\dotsc,u_j)$.
Let $\mathcal{F}^c$ denote the complement of a set $\mathcal{F}$,
and $|\mathcal{F}|$ denotes cardinality of $\mathcal{F}$.
Let $\mathcal{X}$ and $\mathcal{Y}$ be an input alphabet and an output alphabet, respectively.
In this paper, we assume that $\mathcal{X}$ is finite and that $\mathcal{Y}$ is at most countable.
A discrete memoryless channel (DMC) $W$ is defined as a conditional probability distribution $W(y\mid x)$
over $\mathcal{Y}$ where $x\in\mathcal{X}$ and $y\in\mathcal{Y}$.
We write $W:\mathcal{X}\to\mathcal{Y}$ to mean a DMC $W$ with an input alphabet $\mathcal{X}$ and
an output alphabet $\mathcal{Y}$.
Let $q$ be the cardinality of $\mathcal{X}$.
In this paper, the base of the logarithm is $q$ unless otherwise stated.
\begin{definition}
The symmetric capacity of $q$-ary input channel $W:\mathcal{X}\to\mathcal{Y}$ is defined as
\begin{equation*}
I(W) := \sum_{x\in\mathcal{X}}\sum_{y\in\mathcal{Y}} \frac1q W(y\mid x)
\log\frac{W(y\mid x)}{\frac1q \sum_{x'\in\mathcal{X}}W(y\mid x')}.
\end{equation*}
Note that $I(W)\in[0,1]$.
\end{definition}
\begin{definition}
Let $\mathcal{D}_x:=\{y\in\mathcal{Y}\mid W(y\mid x) > W(y\mid x'), \forall x'\in\mathcal{X}, x'\ne x\}$.
The error probability of the maximum-likelihood estimation of the input $x$ 
on the basis of the output $y$ of the channel $W$ is defined as
\begin{equation*}
P_e(W) := \frac1q \sum_{x\in\mathcal{X}} \sum_{y\in\mathcal{D}_x^c} W(y\mid x).
\end{equation*}
\end{definition}
\begin{definition}
The Bhattacharyya parameter of $W$ is defined as
\begin{equation*}
Z(W) := \frac1{q(q-1)} \sum_{\substack{x\in\mathcal{X}, x'\in\mathcal{X},\\ x\ne x'}} Z_{x,x'}(W)
\end{equation*}
where the Bhattacharyya parameter of $W$ between $x$ and $x'$ is defined as
\begin{equation*}
Z_{x,x'}(W) := \sum_{y\in\mathcal{Y}} \sqrt{W(y\mid x)W(y\mid x')}.
\end{equation*}
\end{definition}
The symmetric capacity $I(W)$, the error probability $P_e(W)$, and the Bhattacharyya parameter $Z(W)$
are interrelated as in the following lemmas.
\begin{lemma}
\begin{equation*}
P_e(W)\le (q-1)Z(W).
\end{equation*}
\end{lemma}
\begin{lemma}\label{lem:IZ}\cite{sasoglu2009pad}
\begin{align*}
I(W)&\ge\log\frac{q}{1+(q-1)Z(W)}\\
I(W)&\le\log(q/2)+(\log2)\sqrt{1-Z(W)^2}\\
I(W)&\le 2(q-1)(\log \mathrm{e})\sqrt{1-Z(W)^2}.
\end{align*}
\end{lemma}
\begin{definition}
The maximum and the minimum of the Bhattacharyya parameters between two symbols are defined as
\begin{align*}
Z_\text{max}(W) &:= \max_{x\in\mathcal{X},x'\in\mathcal{X},x\ne x'} Z_{x,x'}(W)\\
Z_\text{min}(W) &:= \min_{x\in\mathcal{X},x'\in\mathcal{X}} Z_{x,x'}(W).
\end{align*}
Let $\sigma:\mathcal{X}\to\mathcal{X}$ be a permutation.
Let $\sigma^i$ denote the $i$th power of $\sigma$.
The average Bhattacharyya parameter of $W$ between $x$ and $x'$ with respect to $\sigma$
is defined as the average of $Z_{z,z'}(W)$ over 
the subset $\{(z,z')=(\sigma^i(x),\sigma^i(x'))\in\mathcal{X}^2\mid i=0,1,\ldots,q!-1\}$ as
\begin{align*}
Z_{x,x'}^{\sigma}(W) &:= \frac1{q!} \sum_{i=0}^{q!-1} Z_{\sigma^i(x),\sigma^i(x')}(W).
\end{align*}
\end{definition}

\section{Channel polarization on $q$-ary DMC induced by non-linear kernel}
We consider a channel transform using a one-to-one onto mapping $g: \mathcal{X}^\ell\to\mathcal{X}^\ell$,
which is called a kernel.
In the previous works~\cite{5075875}, \cite{korada2009pcc}, it is assumed that $q=2$ and that $g$ is linear.
In~\cite{sasoglu2009pad}, $\mathcal{X}$ is arbitrary but $g$ is restricted.
In this paper, $\mathcal{X}$ and $g$ are arbitrary.
\begin{definition}
Let $W:\mathcal{X}\to\mathcal{Y}$ be a DMC.
Let $W^\ell :\mathcal{X}^\ell\to\mathcal{Y}^\ell$,
$W^{(i)}:\mathcal{X}\to\mathcal{Y}^\ell\times\mathcal{X}^{i-1}$, and $W_{u_0^{i-1}}^{(i)}: \mathcal{X}\to\mathcal{Y}^\ell$ be defined as
DMCs with transition probabilities
\begin{align*}
W^\ell(y_0^{\ell-1}\mid x_0^{\ell-1}) &:=\prod_{i=0}^{\ell-1} W(y_i\mid x_i)\\
W^{(i)}(y_0^{{\ell-1}},u_0^{i-1}\mid u_i) &:= \frac1{q^{\ell-1}}
\sum_{u_{i+1}^{\ell-1}} W^{\ell}(y_0^{{\ell-1}}\mid g(u_0^{{\ell-1}}))\\
W_{u_0^{i-1}}^{(i)}(y_0^{{\ell-1}}\mid u_i) &:= \frac1{q^{{\ell}-i-1}}
\sum_{u_{i+1}^{\ell-1}} W^{{\ell}}(y_0^{{\ell-1}}\mid g(u_0^{{\ell-1}})).
\end{align*}
\end{definition}
\begin{definition}
Let $\{B_i\}_{i=0,1,\dotsc}$ be independent random variables such that
$B_i = k$
with probability  $\frac1\ell$,
for each $k=0,\dotsc,\ell-1$.
\end{definition}

In probabilistic channel transform $W\to W^{(B_i)}$, expectation of the symmetric capacity is invariant due to the chain rule for mutual information.
The following lemma is a consequence of the martingale convergence theorem.
\begin{lemma}\label{lem:mc}
There exists a random variable $I_\infty$ such that $I(W^{(B_0)\dotsm(B_n)})$ converges to $I_\infty$
almost surely as $n\to\infty$.
\end{lemma}

When $q=2$ and $g(u_0^1) = (u_0+u_1,u_1)$, Ar{\i}kan showed that $P(I_\infty\in\{0,1\})=1$~\cite{5075875}.
This result is called channel polarization phenomenon since subchannels polarize to noiseless channels and pure noise channels.
Korada, {\c S}a{\c s}o{\u g}lu, and Urbanke consider channel polarization phenomenon when $q=2$ and $g$ is linear~\cite{korada2009pcc}.

From Lemma~\ref{lem:IZ}, $I(W)$ is close to 0 and 1 when $Z(W)$ is close to 1 and 0, respectively.
Hence, it would be sufficient to prove channel polarization
if one can show that $Z(W^{(B_1)\dotsm(B_n)})$ converges to $Z_\infty\in\{0,1\}$ almost surely.
Here we instead show a weaker version of the above property in the following lemma and its corollary.
\begin{lemma}\label{lem:polar}
Let $\{\mathcal{Y}_n\}_{n\in\mathbb{N}}$ be a sequence of discrete sets.
Let $\{W_n:\mathcal{X}\to\mathcal{Y}_n\}_{n\in\mathbb{N}}$ be a sequence of $q$-ary DMCs.
Let $\sigma$ and $\tau$ be permutations on $\mathcal{X}$.
Let
\begin{equation*}
W_n'(y_1, y_2 \mid x) =  W_n(y_1\mid \sigma(x))W_n(y_2\mid \tau(x))
\end{equation*}
where $W_n:\mathcal{X}\to\mathcal{Y}_n$, $W_n' : \mathcal{X}\to\mathcal{Y}_n^2$.
Assume $\lim_{n\to\infty} I(W_n')-I(W_n)=0$.
Then, for any $\delta\in(0,1/2)$, there exists $m$ such that $Z^{\tau\sigma^{-1}}_{x,x'}(W_n) \notin (\delta,1-\delta)$
for any $x\in\mathcal{X}$, $x'\in\mathcal{X}$ and $n\ge m$.
\end{lemma}
\begin{IEEEproof}
Let $Z$, $Y_1$ and $Y_2$ be random variables which take values on $\mathcal{X}$, $\mathcal{Y}_n$ and $\mathcal{Y}_n$, respectively,
and jointly obey the distribution
\begin{multline*}
P_n(Z=z,\,Y_1=y_1,Y_2=y_2)\\
=\frac1q W_n(y_1\mid \sigma(z))W_n(y_2\mid \tau(z)).
\end{multline*}
Since $I(W_n') = I(Z;Y_1,Y_2)$ and $I(W_n) = I(Z;Y_1)$,
\begin{equation*}
I(Z;Y_1,Y_2)-I(Z;Y_1)
=I(Z;Y_2\mid Y_1)
\end{equation*}
tends to 0 by the assumption.
Since the mutual information is lower bounded by the cut-off rate,
one obtains
\begin{align*}
&I(Z;Y_2\mid Y_1)
\ge
-\log \sum_{y_1\in\mathcal{Y}_n,y_2\in\mathcal{Y}_n} P_n(Y_1=y_1)\nonumber\\
&\quad\times\Bigg[\sum_{z\in\mathcal{X}} P_n(Z=z\mid Y_1=y_1)\nonumber\\
&\hspace{9em}\times\sqrt{P_n(Y_2=y_2\mid Z=z, Y_1=y_1)}\Bigg]^2\nonumber\\
&=
-\log \sum_{y_1\in\mathcal{Y}_n,z\in\mathcal{X},x\in\mathcal{X}} P_n(Y_1=y_1)P_n(Z=z\mid Y_1=y_1)\nonumber\\
&\quad\times P_n(Z=x\mid Y_1=y_1)
 Z_{\tau(z),\tau(x)}(W_n)\nonumber\\
&=
-\log \sum_{y_1\in\mathcal{Y}_n,z\in\mathcal{X},x\in\mathcal{X}} q_n(y_1,z,x)
 Z_{\tau(\sigma^{-1}(z)),\tau(\sigma^{-1}(x))}(W_n)
\end{align*}
where
\begin{multline*}
q_n(y_1,z,x)
 := P_n(Y_1=y_1)\\
\times P_n(Z=\sigma^{-1}(z)\mid Y_1=y_1)
 P_n(Z=\sigma^{-1}(x)\mid Y_1=y_1).
\end{multline*}
Since
\begin{align*}
&\sum_{y_1\in\mathcal{Y}} q_n(y_1,z,x) = \sum_{y_1\in\mathcal{Y}}P_n(Y_1=y_1)\\
&\quad\times 
\left(\sqrt{P_n(Z=\sigma^{-1}(z)\mid Y_1=y_1)
 P_n(Z=\sigma^{-1}(x)\mid Y_1=y_1)}\right)^2\\
&\ge
\bigg(\sum_{y_1\in\mathcal{Y}}P_n(Y_1=y_1)\\
&\quad\times 
\sqrt{P_n(Z=\sigma^{-1}(z)\mid Y_1=y_1)
 P_n(Z=\sigma^{-1}(x)\mid Y_1=y_1)}\bigg)^2\\
&=\frac1{q^2} Z_{z,x}(W_n)^2
\end{align*}
it holds
\begin{align*}
&I(Z;Y_2\mid Y_1)
\ge
-\log \bigg[1- \\
&\frac1{q^2}\sum_{z\in\mathcal{X},x\in\mathcal{X}} Z_{z,x}(W_n)^2
(1
- Z_{\tau(\sigma^{-1}(z)),\tau(\sigma^{-1}(x))}(W_n))\bigg].
\end{align*}
The convergence of $I(Z;Y_2\mid Y_1)$ to 0 implies that
\begin{equation*}
Z_{z,x}(W_n)^2
(1
- Z_{\tau(\sigma^{-1}(z)),\tau(\sigma^{-1}(x))}(W_n))
\end{equation*}
converges to 0 for any $(z,x)\in\mathcal{X}^2$.
It consequently implies that
for any $\delta\in(0,1/2)$, there exists $m$ such that $Z^{\tau\sigma^{-1}}_{x,x'}(W_n) \notin (\delta,1-\delta)$
for any $x\in\mathcal{X}$, $x'\in\mathcal{X}$ and $n\ge m$.
\end{IEEEproof}

Using Lemma~\ref{lem:polar},
one can obtain a partial result of the channel polarization as follows.
\begin{corollary}
Assume that 
there exists $u_0^{\ell-2}\in\mathcal{X}^{\ell-1}$, 
$(i,j)\in \{0,1,\dotsc,\ell-1\}^2$
and permutations $\sigma$ and $\tau$ on $\mathcal{X}$ such that
$i$-th element of $g(u_0^{\ell-1})$ and $j$-th element of $g(u_0^{\ell-1})$ are $\sigma(u_{\ell-1})$ and $\tau(u_{\ell-1})$, respectively,
and such that for any $v_0^{\ell-2}\ne u_0^{\ell-2}\in\mathcal{X}^{\ell-1}$
there exists $m\in\{0,1,\dotsc,\ell-1\}$ and a permutation $\mu$ on $\mathcal{X}$ such that
$m$-th element of $g(v_0^{\ell-1})$ is $\mu(v_{\ell-1})$.
Then, for almost every sequence $b_1,\dotsc,b_n,\dotsc$ of $0,\dotsc,\ell-1$, and for any $\delta\in(0,1/2)$,
there exists $m$ such that $Z^{\tau\sigma^{-1}}_{x,x'}(W^{(b_1)\dotsm(b_n)}) \notin (\delta,1-\delta)$
for any $x\in\mathcal{X}$, $x'\in\mathcal{X}$ and $n\ge m$.
\end{corollary}
\begin{proof}
Since $I(W^{(B_1)\dotsm(B_n)})$ converges to $I_\infty$ almost surely,
$|I(W^{(B_1)\dotsm(B_{n})(\ell-1)})-I(W^{(B_1)\dotsm(B_n)})|$ has to converge to 0 almost surely.
Let $U_0^{\ell-1}$ and $Y_0^{\ell-1}$ denote random variables ranging over $\mathcal{X}^{\ell}$ and $\mathcal{Y}^{\ell}$,
and obeying the distribution
\begin{equation*}
P(U_0^{i}=u_0^{\ell-1}, Y_0^{\ell-1}=y_0^{\ell-1}) = \frac1q W^{(\ell-1)}(y_0^{\ell-1},u_0^{\ell-2}\mid u_{\ell-1}).
\end{equation*}
Then, it holds 
\begin{align*}
I(W^{(\ell-1)}) &= I(Y_0^{\ell-1}, U_0^{\ell-2}; U_{\ell-1})\\
&= I(Y_0^{\ell-1}; U_{\ell-1}\mid U_0^{\ell-2})\\
&=\sum_{u_0^{\ell-2}}\frac1{q^{\ell-1}} I(Y_0^{\ell-1}; U_{\ell-1}\mid U_0^{\ell-2}=u_0^{\ell-2}).
\end{align*}
From the assumption,
$I(Y_0^{\ell-1}; U_{\ell-1}\mid U_0^{\ell-2}=u_0^{\ell-2})\ge I(W)$ for all $u_0^{\ell-2}\in\mathcal{X}^{\ell-1}$.
Hence,
$I(W^{(B_1)\dotsm(B_{n})'})-I(W^{(B_1)\dotsm(B_n)})$ has to converge to 0 almost surely.
By applying Lemma~\ref{lem:polar}, one obtains the result.
\end{proof}
When $q=2$,
since $Z(W)=Z_{0,1}(W)$,
this corollary immediately implies the channel polarization phenomenon,
although it is not sufficient for general $q\ne 2$.
Note that in this derivation one does not use extra conditions e.g., symmetricity of DMC, linearity of a kernel.
%

If a kernel is linear, a more detailed condition is obtained.
\begin{definition}
Assume $(\mathcal{X}, +, \cdot)$ be a commutative ring. A kernel $g:\mathcal{X}^\ell\to\mathcal{X}^\ell$ is said to be linear if $g(ax+bz) = ag(x) + bg(z)$
for all $a\in\mathcal{X}$, $b\in\mathcal{X}$, $x\in\mathcal{X}^\ell$, and $z\in\mathcal{X}^\ell$.
\end{definition}

If $g$ is linear, $g$ can be represented by a square matrix $G$ such that $g(u_0^{\ell-1}) = u_0^{\ell-1}G$.
Let $U_0^{\ell-1}$, $X_0^{\ell-1}$ and $Y_0^{\ell-1}$ denote random variables taking values on
$\mathcal{X}^\ell$, $\mathcal{X}^\ell$ and $\mathcal{Y}^\ell$, respectively, and obeying distribution
\begin{multline*}
P(U_0^{\ell-1}=u_0^{\ell-1}, X_0^{\ell-1} = x_0^{\ell-1}, Y_0^{\ell-1} = y_0^{\ell-1}) \\
= \frac1{2^\ell} W^\ell\left(y_0^{\ell-1}\mid u_0^{\ell-1}G\right)
\mathbb{I}\{x_0^{\ell-1}V=u_0^{\ell-1}\}
\end{multline*}
where $V$ denotes an $\ell\times\ell$ full-rank upper triangle matrix.
There exists a one-to-one correspondence between $X_0^i$ and $U_0^i$ for all $i\in\{0,\dotsc,\ell-1\}$.
Hence, statistical properties of $W^{(i)}$ are invariant under an operation $G\to VG$.
Further, a permutation of columns of $G$ does not change statistical properties of $W^{(i)}$ either.
Since any full-rank matrix can be decomposed to the form $VLP$ where $V$, $L$, and $P$ are
upper triangle, lower triangle, and permutation matrices,
without loss of generality we assume that $G$ is a lower triangle matrix and
that $G_{kk} = 1 $ where
$k\in\{0,\dotsc,\ell-1\}$ is the largest number such that the number of non-zero elements in $k$-th row of $G$ is greater than 1,
and where $G_{ij}$ denotes $(i,j)$ element of $G$.
\begin{theorem}\label{thm:prime}
Assume that $\mathcal{X}$ is a field of prime cardinality, and that linear kernel $G$ is not diagonal.
Then, $P(I_\infty\in\{0,1\})=1$.
\end{theorem}
\begin{IEEEproof}
It holds
\begin{multline*}
W^{(k)}(y_0^{\ell-1}, u_0^{k-1} \mid u_k) = \frac1{q^{\ell-1}}
 \prod_{j=k+1}^{\ell-1} \left(\sum_{x\in\mathcal{X}} W(y_j \mid x)\right)\\
\times \prod_{j\in S_0} W(y_j\mid x_j) \prod_{j\in S_1} W(y_j\mid G_{kj} u_k + x_j)
\end{multline*}
where $S_0 := \{j \in\{0,\dotsc,\ell-1\}\mid G_{kj} = 0\}$,
$S_1 := \{j \in\{0,\dotsc,\ell-1\}\mid G_{kj} \ne 0\}$, and
$x_j$ is $j$-th element of $(u_0^{k-1}, 0_k^{\ell-1})G$ where $0_k^{\ell-1}$ is all-zero vector of length $\ell-k$.
Let $m\in\{0,\dotsc,k-1\}$ be such that $G_{km}\ne 0$.
Since each $u_0^{k-1}$ occurs with positive probability $1/q^{k}$,
we can apply Lemma~\ref{lem:polar} with $\sigma(x) = x$ and $\tau(x) = G_{km}x + z$ for arbitrary $z\in\mathcal{X}$.
Hence, for sufficiently large $n$, $Z_{x,x'}^{\mu}(W^{(B_1)\dotsm(B_n)})$ is close to 0 or 1 almost surely
where $\mu(x) = G_{km}^ix +z$ for all $i\in\{0,\dotsc,q-2\}$ and $z\in\mathcal{X}$.
Since $q$ is a prime,
when $\mu_0(z)=z+x'-x$ for $x\ne x'$,
$Z_{x,x'}^{\mu_0}(W^{(B_1)\dotsm(B_n)})$ is close to 0 or 1 if and only if $Z(W^{(B_1)\dotsm(B_n)})$ is close to 0 or 1, respectively.
\end{IEEEproof}
This result is a simple generalization of the special case considered by {\c S}a{\c s}o{\u g}lu, Telatar, and Ar{\i}kan~\cite{sasoglu2009pad}.
For a prime power $q$ and a finite field $\mathcal{X}$, 
we show a sufficient condition for channel polarization in the following corollary. 
\begin{corollary}\label{cor:ff}
Assume that $\mathcal{X}$ is a field and that a linear kernel $G$ is not diagonal.
If there exists $j\in\{0,\dotsc,k-1\}$ such that $G_{kj}$ is a primitive element.
Then, $P(I_\infty\in\{0,1\})=1$.
\end{corollary}
\begin{IEEEproof}
By applying Lemma~\ref{lem:polar}, one sees that
for almost every sequence $b_1,\dotsc,b_n,\dotsc$ of $0,\dotsc,\ell-1$, and for any $\delta\in(0,1/2)$,
there exists $m$ such that $Z^{\sigma}_{x,x'}(W^{(b_1)\dotsm(b_n)}) \notin (\delta,1-\delta)$
for any $x\in\mathcal{X}$, $x'\in\mathcal{X}$ and $n\ge m$
%
where $\sigma(x) = G_{kj}x + z$ for arbitrary $z\in\mathcal{X}$.
It suffices to show that
for any $x\in\mathcal{X}$ and $x'\in\mathcal{X}$, $x\ne x'$
$Z_{x,x'}(W^{(B_1)\dotsm(B_n)})$ is close to 1 if and only if
$Z(W^{(B_1)\dotsm(B_n)})$ is close to 1.
When $Z_{x,x'}(W^{(B_1)\dotsm(B_n)})$ is close to 1,
$Z_{0,G_{kj}(x'-x)}(W^{(B_1)\dotsm(B_n)})$ is close to 1. 
Hence,
$Z_{0,G_{kj}^i(x'-x)}(W^{(B_1)\dotsm(B_n)})$ is close to 1 for any $i\in\{0,\dotsc,q-2\}$.
Since $G_{kj}$ is a primitive element,
$Z_{0,x}(W^{(B_1)\dotsm(B_n)})$ is close to 1 for any $x\in\mathcal{X}$.
It completes the proof.
\end{IEEEproof}
In~\cite{5351487}, it is shown that
the channel polarization phenomenon occurs
by using a random kernel in which $G_{kj}$ is chosen uniformly from nonzero elements.
Corollary~\ref{cor:ff} says that a deterministic primitive element $G_{kj}$ is sufficient for the channel polarization phenomenon.

\section{Speed of polarization}
Ar{\i}kan and Telatar showed the speed of polarization~\cite{5205856}.
Korada, {\c S}a{\c s}o{\u g}lu, and Urbanke generalized it to any binary linear kernels~\cite{korada2009pcc}.
\begin{proposition}\label{prop:X}
Let $\{\hat{X}_n\in(0,1)\}_{n\in\mathbb{N}}$ be a random process satisfying the following properties.
\begin{enumerate}
\item $\hat{X}_n$ converges to $\hat{X}_\infty$ almost surely.
\item $\hat{X}_{n+1} \le \hat{c}\hat{X}_n^{\hat{D}_n}$ where $\{\hat{D}_n\ge 1\}_{n\in\mathbb{N}}$ are
 independent and identically distributed random variables,
and $\hat{c}$ is a constant.
\end{enumerate}
Then,
\begin{equation*}
\lim_{n\to\infty} P(\hat{X}_n < 2^{-2^{\beta n}}) = P(\hat{X}_\infty = 0)
\end{equation*}
for $\beta < \mathbb{E}[\log_2 \hat{D}_1]$ where $\mathbb{E}[\cdot]$ denotes an expectation.
Similarly,
let $\{\check{X}_n\in(0,1)\}_{n\in\mathbb{N}}$ be a random process satisfying the following properties.
\begin{enumerate}
\item $\check{X}_n$ converges to $\check{X}_\infty$ almost surely.
\item $\check{X}_{n+1} \ge \check{c}\check{X}_n^{\check{D}_n}$ where $\{\check{D}_n\ge 1\}_{n\in\mathbb{N}}$ are
 independent and identically distributed random variables,
and $\check{c}$ is a constant.
\end{enumerate}
Then,
\begin{equation*}
\lim_{n\to\infty} P(\check{X}_n < 2^{-2^{\beta n}}) = 0
\end{equation*}
for $\beta > \mathbb{E}[\log_2 \check{D}_1]$.
\end{proposition}
Note that the above proposition can straightforwardly be extended to 
include the rate dependence~\cite{tanaka2010rre}.

In order to apply Proposition~\ref{prop:X} to $Z_\text{max}(W^{(B_1)\dotsm(B_n)})$ and $Z_\text{min}(W^{(B_1)\dotsm(B_n)})$
as $\hat{X}_n$ and $\check{X}_n$, respectively,
the second conditions have to be proven.
In the argument of~\cite{korada2009pcc}, partial distance of a kernel corresponds to
the random variables $\hat{D}_n$ and $\check{D}_n$ in Proposition~\ref{prop:X}.
\begin{definition}
Partial distance of a kernel $g: \mathcal{X}^\ell\to\mathcal{X}^\ell$ is defined as
\begin{multline*}
D_{x,x'}^{(i)}(u_0^{i-1})\\
:= \min_{v_{i+1}^{\ell-1}, w_{i+1}^{\ell-1}} 
d(g(u_0^{i-1},x,v_{i+1}^{\ell-1}),g(u_0^{i-1},x',w_{i+1}^{\ell-1}))
\end{multline*}
where $d(a,b)$ denotes the Hamming distance between $a\in\mathcal{X}^\ell$ and $b\in\mathcal{X}^\ell$.
\end{definition}
We also use the following quantities.
\begin{align*}
D_{x,x'}^{(i)} &:= \min_{u_0^{i-1}} D_{x,x'}^{(i)}(u_0^{i-1})\\
D_\text{max}^{(i)} &:= \max_{x\in\mathcal{X},x'\in\mathcal{X}} D_{x,x'}^{(i)}\\
D_\text{min}^{(i)} &:= \min_{\substack{x\in\mathcal{X},x'\in\mathcal{X}\\x\ne x'}} D_{x,x'}^{(i)}.
\end{align*}
When $g$ is linear, $D^{(i)}_{x,x'}(u_0^{i-1})$ does not depend on $x$, $x'$ or $u_0^{i-1}$,
in which case
we will use the notation $D^{(i)}$ instead of $D^{(i)}_{x,x'}(u_0^{i-1})$.

From Lemma~\ref{lem:ZD} in the appendix, the following lemma is obtained.
\begin{lemma}
For $i\in\{0,\dotsc,\ell-1\}$,
\begin{equation*}
\frac1{q^{2\ell-2-i}}Z_\text{min}(W)^{D_{x,x'}^{(i)}} \le Z_{x,x'}(W^{(i)}_{\ell}) \le q^{\ell-1-i}Z_\text{max}(W)^{D_{x,x'}^{(i)}}
\end{equation*}
\end{lemma}

\begin{corollary}\label{cor:Zminmax}
For $i\in\{0,\dotsc,\ell-1\}$,
\begin{align*}
Z_\text{max}(W^{(i)}) &\le q^{\ell-1-i}Z_\text{max}(W)^{D^{(i)}_\text{min}}\\
\frac1{q^{2\ell-2-i}}Z_\text{min}(W)^{D^{(i)}_\text{max}} &\le Z_\text{min}(W^{(i)}).
\end{align*}
\end{corollary}
From Proposition~\ref{prop:X} and Corollary~\ref{cor:Zminmax}, the following theorem is obtained.
\begin{theorem}
Assume $P(I_\infty(W)\in\{0,1\})=1$.
It holds
\begin{equation*}
\lim_{n\to\infty} P(Z(W^{(B_1)\dotsc(B_n)}) < 2^{-\ell^{\beta n}}) = I(W)
\end{equation*}
for $\beta < (1/\ell) \sum_i \log_\ell D^{(i)}_\text{min}$.

When $Z_\text{min}(W)>0$,
\begin{equation*}
\lim_{n\to\infty} P(Z(W^{(B_1)\dotsc(B_n)}) < 2^{-\ell^{\beta n}}) = 0
\end{equation*}
for $\beta > (1/\ell) \sum_i \log_\ell D^{(i)}_\text{max}$.
\end{theorem}

When $g$ is a linear kernel represented by a square matrix $G$, 
$(1/\ell)\sum_i \log_\ell D^{(i)}$ is called the exponent of $G$~\cite{korada2009pcc}.

\begin{example}
Assume that $\mathcal{X}$ is a field and that $\alpha\in\mathcal{X}$ is a primitive element.
For a non-zero element $\gamma\in\mathcal{X}$, let
\begin{equation*}
G_\text{RS}(q)=
\begin{bmatrix}
1&1&\dotsc&1&1&0\\
\alpha^{(q-2)(q-2)}&\alpha^{(q-3)(q-2)}&\dotsc&\alpha^{q-2}&1&0\\
\alpha^{(q-2)(q-3)}&\alpha^{(q-3)(q-3)}&\dotsc&\alpha^{q-3}&1&0\\
\vdots&\vdots&\dotsc&\vdots&\vdots&\vdots\\
\alpha^{q-2}&\alpha^{q-3}&\dotsc& \alpha&1&0\\
1&1& \dotsc& 1&1&\gamma\\
\end{bmatrix}.
\end{equation*}
Since $G_\text{RS}(2)=\begin{bmatrix}1&0\\1&1\end{bmatrix}$,
$G_\text{RS}(q)$ can be regarded as a generalization of Ar{\i}kan's original matrix.
The relation between binary polar codes and binary Reed-Muller codes~\cite{5075875}
also holds for $q$-ary polar codes using $G_\text{RS}(q)$ and $q$-ary Reed-Muller codes.
From Theorem~\ref{thm:prime}, the channel polarization phenomenon occurs on $G_\text{RS}(q)$ for any $\gamma\ne0$ when $q$ is a prime.
When $\gamma$ is a primitive element, from Corollary~\ref{cor:ff},
the channel polarization phenomenon occurs on $G_\text{RS}(q)$ for any prime power $q$.
We call $G_\text{RS}(q)$ the Reed-Solomon kernel
%
since the submatrix which consists of $i$-th row to $(q-1)$-th row of $G_\text{RS}(q)$
is a generator matrix of a generalized Reed-Solomon code,
which is a maximum distance separable code i.e., $D^{(i)} = i+1$.
Hence, the exponent of $G_\text{RS}(q)$ is $\frac1\ell \sum_i \log_\ell (i+1)$ where $\ell=q$.
Since
\begin{equation*}
\frac1\ell \sum_{i=0}^{\ell-1}\log_\ell (i+1) 
\ge
\frac1{\ell\log_\mathrm{e}\ell}\int_1^\ell \log_\mathrm{e} x \mathrm{d}x
= 1- \frac{\ell-1}{\ell\log_\mathrm{e}\ell}
\end{equation*}
the exponent of the Reed-Solomon kernel tends to 1 as $\ell=q$ tends to infinity.
When $q=2^2$, the exponent of the Reed-Solomon kernel is $\log_\mathrm{e} 24 / (4\log_\mathrm{e} 4)\approx 0.57312$.
In Ar{\i}kan's original work, the exponent of the $2\times 2$ matrix is $0.5$~\cite{5205856}.
In~\cite{korada2009pcc}, Korada, {\c S}a{\c s}o{\u g}lu, and Urbanke showed that by using large kernels, the exponent can be improved,
and found a matrix of size 16 whose exponent is about 0.51828.
The above-mentioned Reed-Solomon kernel with $q=2^2$ is reasonably small and simple but has
a larger exponent than binary linear kernels of small size.
This demonstrates the usefulness of considering $q$-ary rather than binary channels.
For $q$-ary DMC where $q$ is not a prime, it can be decomposed to subchannels of input sizes of prime numbers~\cite{5351487}
by using the method of multilevel coding~\cite{1055718}.
The above example shows that 
when $q$ is a power of a prime,
without the decomposition of $q$-ary DMC,
asymptotically better coding scheme can be constructed 
by using $q$-ary polar codes with $G_\text{RS}(q)$.
\end{example}

\section{Conclusion}
The channel polarization phenomenon on $q$-ary channels has been considered.
We give several sufficient conditions on kernels under which the channel polarization phenomenon occurs.
We also show an explicit construction with a $q$-ary linear kernel $G_\text{RS}(q)$ for $q$ being a power of a prime.
The exponent of $G_\text{RS}(q)$ is $\log_\mathrm{e}(q!)/(q\log_\mathrm{e}q)$ which is larger than the exponent of binary matrices of small size
even if $q=4$.
%
Our discussion includes channel polarization on non-linear kernels as well.
It is known that non-linear binary codes may have a larger minimum distance than linear binary codes,
e.g. the Nordstrom-Robinson codes~\cite{macwilliams1988tec}.
This implies possibility that there exists a non-linear kernel with a larger exponent than any linear kernel of the same size.

\appendix
\begin{lemma}\label{lem:ZD}
\begin{multline*}
\frac1{q^{2(\ell-1-i)}}Z_\text{min}(W)^{D_{x,x'}^{(i)}(u_0^{i-1})}\\
\le Z_{x,x'}(W^{(i)}_{u_0^{i-1}}) \le q^{\ell-1-i}Z_\text{max}(W)^{D_{x,x'}^{(i)}(u_0^{i-1})}
\end{multline*}
\end{lemma}
\enlargethispage{-10em}
\begin{IEEEproof}
For the second inequality, one has
\begin{align*}
&Z_{x,x'}(W_{ u_0^{i-1}}^{(i)}) = \sum_{y_0^{\ell-1}} \sqrt{W_{u_0^{i-1}}^{(i)}(y_0^{\ell-1}\mid x)W_{u_0^{i-1}}^{(i)}(y_0^{\ell-1}\mid x')}\\
&=q^{i}\sum_{y_0^{\ell-1}} \sqrt{W^{(i)}(y_0^{\ell-1},u_0^{i-1}\mid x)W^{(i)}(y_0^{\ell-1},u_0^{i-1}\mid x')}\\
&=\frac1{q^{\ell-1-i}}\sum_{y_0^{\ell-1}}
\Bigg(\sum_{v_{i+1}^{\ell-1},w_{i+1}^{\ell-1}}\\
& W^{\ell}(y_0^{\ell-1}\mid u_0^{i-1}, x, v_{i+1}^{\ell-1})
W^{\ell}(y_0^{\ell-1}\mid u_0^{i-1}, x', w_{i+1}^{\ell-1})\Bigg)^\frac12\\
&\le \frac1{q^{\ell-1-i}}\sum_{y_0^{\ell-1}} \sum_{v_{i+1}^{\ell-1},w_{i+1}^{\ell-1}}\\
&\quad\sqrt{W^{\ell}(y_0^{\ell-1}\mid u_0^{i-1}, x, v_{i+1}^{\ell-1})W^{\ell}(y_0^{\ell-1}\mid u_0^{i-1}, x', w_{i+1}^{\ell-1})}\\
&\le \frac1{q^{\ell-1-i}}
\sum_{v_{i+1}^{\ell-1},w_{i+1}^{\ell-1}}
Z_\text{max}(W)^{D_{x,x'}^{(i)}(u_0^{i-1})}\\
&= q^{\ell-1-i} Z_\text{max}(W)^{D_{x,x'}^{(i)}(u_0^{i-1})}.
\end{align*}
The first inequality is obtained as follows.
\begin{align*}
&Z_{x,x'}(W_{ u_0^{i-1}}^{(i)}) = 
\sum_{y_0^{\ell-1}} \sqrt{W_{u_0^{i-1}}^{(i)}(y_0^{\ell-1}\mid x)W_{u_0^{i-1}}^{(i)}(y_0^{\ell-1}\mid x')}\\
&=q^{i}\sum_{y_0^{\ell-1}}
\sqrt{W^{(i)}(y_0^{\ell-1},u_0^{i-1}\mid x)W^{(i)}(y_0^{\ell-1},u_0^{i-1}\mid x')}\\
&=\sum_{y_0^{\ell-1}}
\Biggl(\sum_{v_{i+1}^{\ell-1},w_{i+1}^{\ell-1}}
\frac1{q^{2(\ell-1-i)}}\\
&\quad\times W^{\ell}(y_0^{\ell-1}\mid u_0^{i-1}, x, v_{i+1}^{\ell-1})W^{\ell}(y_0^{\ell-1}\mid u_0^{i-1}, x', w_{i+1}^{\ell-1})\Biggr)^{\frac12}\\
&\ge\sum_{y_0^{\ell-1}}
\sum_{v_{i+1}^{\ell-1},w_{i+1}^{\ell-1}} \frac1{q^{2(\ell-1-i)}}\\
&\quad \times\sqrt{W^{\ell}(y_0^{\ell-1}\mid u_0^{i-1}, x, v_{i+1}^{\ell-1})W^{\ell}(y_0^{\ell-1}\mid u_0^{i-1}, x', w_{i+1}^{\ell-1})}\\
&\ge  \frac1{q^{2(\ell-1-i)}} 
Z_\text{min}(W)^{D_{x,x'}^{(i)}(u_0^{i-1})}.
\end{align*}
\end{IEEEproof}

\section*{Acknowledgment}
TT acknowledges support of the Grant-in-Aid for Scientific Research
on Priority Areas (No.~18079010), MEXT, Japan.

\bibliographystyle{IEEEtran}
\bibliography{IEEEabrv,ldpc}

\begin{thebibliography}{1}
\providecommand{\url}[1]{#1}
\csname url@samestyle\endcsname
\providecommand{\newblock}{\relax}
\providecommand{\bibinfo}[2]{#2}
\providecommand{\BIBentrySTDinterwordspacing}{\spaceskip=0pt\relax}
\providecommand{\BIBentryALTinterwordstretchfactor}{4}
\providecommand{\BIBentryALTinterwordspacing}{\spaceskip=\fontdimen2\font plus
\BIBentryALTinterwordstretchfactor\fontdimen3\font minus
  \fontdimen4\font\relax}
\providecommand{\BIBforeignlanguage}[2]{{%
\expandafter\ifx\csname l@#1\endcsname\relax
\typeout{** WARNING: IEEEtran.bst: No hyphenation pattern has been}%
\typeout{** loaded for the language `#1'. Using the pattern for}%
\typeout{** the default language instead.}%
\else
\language=\csname l@#1\endcsname
\fi
#2}}
\providecommand{\BIBdecl}{\relax}
\BIBdecl

\bibitem{5075875}
E.~Ar{\i}kan, ``Channel polarization: A method for constructing
  capacity-achieving codes for symmetric binary-input memoryless channels,''
  \emph{{IEEE} Trans. Inf. Theory}, vol.~55, no.~7, pp. 3051--3073, July 2009.

\bibitem{korada2009pco}
\BIBentryALTinterwordspacing
S.~Korada and R.~Urbanke, ``{Polar codes are optimal for lossy source
  coding},'' 2009. [Online]. Available: \url{http://arxiv.org/abs/0903.0307}
\BIBentrySTDinterwordspacing

\bibitem{5205856}
E.~Ar{\i}kan and E.~Telatar, ``On the rate of channel polarization,'' in
  \emph{Proc. 2009 IEEE International Symposium on Information Theory}, June
  28-July 3 2009, pp. 1493--1495.

\bibitem{tanaka2010rre}
\BIBentryALTinterwordspacing
T.~Tanaka and R.~Mori, ``{Refined rate of channel polarization},'' 2010.
  [Online]. Available: \url{http://arxiv.org/abs/1001.2067}
\BIBentrySTDinterwordspacing

\bibitem{korada2009pcc}
\BIBentryALTinterwordspacing
S.~Korada, E.~{\c S}a{\c s}o{\u g}lu, and R.~Urbanke, ``{Polar codes:
  characterization of exponent, bounds, and constructions},'' 2009. [Online].
  Available: \url{http://arxiv.org/abs/0901.0536}
\BIBentrySTDinterwordspacing

\bibitem{sasoglu2009pad}
\BIBentryALTinterwordspacing
E.~{\c S}a{\c s}o{\u g}lu, E.~Telatar, and E.~Ar{\i}kan, ``{Polarization for
  arbitrary discrete memoryless channels},'' 2009. [Online]. Available:
  \url{http://arxiv.org/abs/0908.0302}
\BIBentrySTDinterwordspacing

\bibitem{5351487}
E.~Sasoglu, E.~Telatar, and E.~Arikan, ``Polarization for arbitrary discrete
  memoryless channels,'' in \emph{Proc.\ 2009 IEEE Information Theory Workshop,
  Taormina, Italy}, 11--16 Oct. 2009, pp. 144--148.

\bibitem{1055718}
H.~Imai and S.~Hirakawa, ``A new multilevel coding method using
  error-correcting codes,'' \emph{Information Theory, IEEE Transactions on},
  vol.~23, no.~3, pp. 371--377, may 1977.

\bibitem{macwilliams1988tec}
F.~MacWilliams and N.~Sloane, \emph{{The Theory of Error-Correcting
  Codes}}.\hskip 1em plus 0.5em minus 0.4em\relax North-Holland Amsterdam,
  1977.

\end{thebibliography}

\end{document}